\documentclass[reqno, 11pt]{amsart}
\usepackage[utf8]{inputenc}
\usepackage[all]{xy}
\usepackage{enumerate}
\usepackage[english]{babel}
\usepackage{graphicx,xcolor} 
\usepackage[colorlinks=true, linkcolor=blue, citecolor=green, urlcolor=purple]{hyperref}
\usepackage{setspace}
\usepackage{cancel}
\usepackage[left=2.5cm, right=2.5cm, bottom=2cm]{geometry}
\usepackage{amsmath, amssymb, amsthm, epsfig, mathtools}
\usepackage{latexsym}
\usepackage{bbm}
\usepackage[nameinlink,noabbrev]{cleveref}
\usepackage[pagewise]{lineno}
\usepackage{comment}

\def\today{\ifcase\month\or
  January\or February\or March\or April\or May\or June\or
  July\or August\or September\or October\or November\or December\fi
  \space\number\day, \number\year}

\newcommand{\R}{\mathbb{R}}
\newcommand{\N}{\mathbb{N}}
\newcommand{\T}{\mathbb{T}}
\newcommand{\Z}{\mathbb{Z}}
\newcommand{\eps}{\varepsilon}
\newcommand{\dis}{\displaystyle}
\renewcommand{\Re}{\text{Re}}

\DeclareMathOperator{\Span}{span}
\DeclareMathOperator{\PWC}{PWC}
\newtheorem{theorem}{Theorem}[section]
\newtheorem{lemma}{Lemma}[section]
\newtheorem{proposition}{Proposition}[section]
\newtheorem{corollary}{Corollary}[section]
\theoremstyle{definition}
\newtheorem{definition}{Definition}[section]

\theoremstyle{remark}
\newtheorem{remark}{Remark}[section]

\crefname{theorem}{Theorem}{Theorems}
\crefname{lemma}{Lemma}{Lemmas}
\crefname{proposition}{Proposition}{Propositions}
\crefname{corollary}{Corollary}{Corollaries}
\crefname{definition}{Definition}{Definitions}
\crefname{example}{Example}{Examples}
\crefname{remark}{Remark}{Remarks}
\crefname{assum}{Hypothesis}{Hypotheses}
\Crefname{section}{Appendix}{Appendix}

\begin{document}

\date{\today}

\title{Bilinear controllability for the linear KdV--Schr\"{o}dinger equation}

\author[R. Buffe]{Rémi Buffe}
\address{Universit\'e de Lorraine, CNRS, IECL, F-54000 Nancy, France.} 
\email{remi.buffe@univ-lorraine.fr}
\author[A. Duca]{Alessandro Duca}
\address{Universit\'e de Lorraine, CNRS, Inria, IECL, F-54000 Nancy, France.}
\email{alessandro.duca@inria.fr}
\author[H. Parada]{Hugo Parada}
\address{Universit\'e de Lorraine, CNRS, Inria, IECL, F-54000 Nancy, France.}
\email{hugo.parada@inria.fr}
\thanks{This project is funded by the Agence Nationale de la Recherche by the QuBiCCS project (ANR-24-CE40-3008)}

\begin{abstract}

We study the controllability of a linear KdV--Schr\"odinger equation on the one-dimensional torus via purely imaginary bilinear controls. Considering controls spanning a suitable finite number of Fourier modes, we prove small-time global approximate controllability in $L^2(\mathbb{T})$. The result holds between any pair of states with the same norm and is obtained via the saturation method by following the idea introduced in \cite{pozzoli2024small}. We first establish small-time controllability for phase multiplications, and then generate transport operators associated with diffeomorphisms of the torus. Finally, we combine these results to recover global approximate controllability. Note that the controllability property holds independently of the Schr\"odinger component of the dynamics, which may in particular be taken to vanish.

\end{abstract}

\maketitle




\section{Introduction}
We consider the following Korteweg-de Vries Schrödinger equation on the torus $\mathbb{T}:=\R/(2\pi\Z)$ in the presence of bilinear controls
\begin{equation}
\label{eq:kdv-sch}
\begin{cases}
\partial_t\psi+\partial_x^3\psi-i\alpha\partial_x^2\psi-i\lambda |\psi|^{2}\psi
+\beta (|\psi|^{2}\psi)_x
+\gamma (|\psi|^{2})_x\psi=i(u\cdot Q)\psi,&  \text{ on } (0,T)\times\T,\\
\psi(0,\cdot)=\psi_0.
\end{cases}
\end{equation}
Here, $\alpha,\lambda,\beta,\gamma\in\R$ are given parameters. In particular, $\alpha$ measures the strength of the Schrödinger dispersion. The vector-valued function $Q:\T\rightarrow \R^q$ with $q\in \N^*$ is fixed and $u:(0,T)\rightarrow \R^q$ is the control. The equation under consideration can be viewed as a controlled Airy–Schrödinger equation, combining Korteweg–de Vries–type third-order dispersion with nonlinear Schrödinger dynamics, this model has been originally introduced by Hasegawa and Kodama in \cite{kodama1985optical,kodama1987nonlinear} to describe pulse propagation in optical fibers. 

\medskip
The parameter $\alpha$ can be set to zero, thus recovering the linear KdV dispersion with complex-valued state in the presence of bilinear action. The KdV equation $\partial_t\psi+\partial_x^3\psi+\psi\partial_x \psi= 0$ is a classical model for the unidirectional propagation of small-amplitude long waves in weakly nonlinear dispersive media. Originally derived by Boussinesq in the 1870s and later rediscovered and popularized by
Korteweg and de~Vries in their 1895 work, the equation arose as an asymptotic model for the
propagation of long, small-amplitude water waves. It subsequently became fundamental in plasma physics, nonlinear lattices, and ion–acoustic wave theory. The discovery of solitons and the inverse scattering transform further stimulated a vast mathematical literature, ranging from integrable systems to well-posedness theory and long-time dynamics.

The mathematical control theory of the KdV equation began with the pioneering works of
Russell and Zhang~\cite{russell1993controllability,russell1996exact}, who established local internal controllability and stabilization in the periodic setting. Since then,
controllability and stabilization of KdV have been investigated extensively. Boundary control problems reveal a striking dichotomy: the system behaves essentially parabolic when the control acts at the left endpoint, and hyperbolic when it acts at the right endpoint. Local exact controllability via boundary controls was proved for non-critical lengths, and later extended to critical domains through refined spectral and moment arguments. 

On periodic domains, the fundamental work of Laurent, Rosier, and Zhang \cite{laurent2010control} established global exact controllability and global exponential stabilization for KdV with internal additive controls. Their analysis relies on the propagation of compactness and regularity in Bourgain spaces for the associated linear system, combined with sophisticated feedback constructions that exploit the nonlinear dispersive structure of the equation.

Despite this extensive progress, comparatively little is known about {\em multiplicative} or {\em bilinear} controls for the KdV-type dynamics. Bilinear control terms of the form
\begin{equation}
\label{eq:real_control}
(u(t)\cdot Q(x))\psi,\end{equation}
arise naturally in models where the medium is modulated or actuated through spatially dependent coefficients rather than external forcing an effect that occurs, for instance, in plasma manipulation or variable–depth water wave channels (see \Cref{app: physical}). In the present work, we consider a purely imaginary bilinear action. This choice is natural in Schrödinger models for optical fibers. In the plasma setting, this term can be interpreted at the envelope level, where plasma inhomogeneities generate phase modulation or local detuning \cite{bussac1985soliton,benisti2016envelope}.  While bilinear controllability has recently been developed for Schrödinger, heat, Burgers, and wave equations through geometric saturation techniques \cite{duca2024bilinear,duca2025small,duca_burger,pozzoli2024small}, analogous results for KdV remain scarce. To the best of our knowledge, the only existing work addressing bilinear KdV-type control is due to Muñoz \cite{munoz2014approximate}, who proved a large--time approximate controllability to solitons result. We also mention the work of Chen \cite{chen2023global} where global approximate controllability was proved by acting with additive controls using saturation techniques. The extension of these bilinear techniques to the KdV equations with a control law of the form \eqref{eq:real_control} remains an open problem and will be addressed in future work. 

In this work, we first focus on linear version of \eqref{eq:kdv-sch}, namely
\begin{equation}
\tag{KdV-Sch}
\label{eq:lkdv-sch}
\begin{cases}
\partial_t\psi+\partial_x^3\psi-i\alpha\partial_x^2\psi=i(u\cdot Q)\psi&  \text{ in } (0,T)\times\T,\\
\psi(0,\cdot)=\psi_0&  \text{ in } \T,
\end{cases}
\end{equation}
where $\alpha\in\R$, $u\in \PWC(0,T;\R^q)$ and the control profiles $Q_{j}:\T \rightarrow \R$ satisfy
\begin{equation}
\label{eq: controls}
\{1, \sin(x), \cos(x), \sin(2x), \cos(2x)\} \subset \Span\{Q_0,\dots,Q_{q-1}\}
\end{equation}
In this setting, the bilinear control term is purely imaginary. Therefore, the full dynamics is conservative in $L^{2}(\T)$. This constraint means that we cannot affect the $L^{2}$-norm of the solution; consequently, our control results are stated for initial data and targets having the same norm. The main result of this work is the following one.

\begin{theorem}[Small-time global $L^2$ approximate controllability]
\label{th: L2 global}
For every $\psi_0,\psi_1\in L^2(\T)$ with $\|\psi_0\|_{L^{2}(\T)}=\|\psi_1\|_{L^{2}(\T)}$ and $\eps>0$, there exists a time $T\in(0,\eps)$ and a control $u\in \PWC(0,T;\R^q)$ such that the corresponding solution of \eqref{eq:lkdv-sch} satisfies
\[\|\psi(T)-\psi_{1}\|_{L^{2}} < \varepsilon .\]
\end{theorem}

\begin{remark}
The previous result holds for all $\alpha \in \R$. In particular, by setting $\alpha=0$ we recover small-time global approximate controllability for the linear complex KdV equation.
\end{remark}

\subsection*{The saturation method}
The \emph{saturation method} is a geometric control strategy designed for bilinear evolution equations and has recently proved effective across a wide range of PDEs. Its core idea is to compensate for the apparent weakness of a finite-dimensional, multiplicative control by exploiting fast oscillations and conjugation effects that accumulate into nontrivial asymptotic dynamics. The method relies on applying large controls on shrinking time intervals so that, after suitable rescaling, the combined effect of the drift and the control operators produces \emph{effective directions} in the state space that are not directly accessible. These directions arise as limits of conjugated flows, and can be explicitly identified with iterated commutators (or Lie brackets) of the generators. When the family of these generated directions spans a dense subspace a property referred to as \emph{saturation}, one obtains a controllability mechanism with remarkable strength. 

The origins of this philosophy trace back to the works of Agrachev and Sarychev on additive
controls for fluid equations. In their seminal papers \cite{agrachev2005navier, agrachev2006controllability}, they established global approximate controllability of the 2D Navier--Stokes and Euler systems by constructing fast oscillating controls that produce, in the limit, all Lie brackets generated by the drift and the forced modes. More recently, saturation techniques have been adapted to \emph{bilinear} control systems. In particular, Nersesyan and Duca \cite{duca2024bilinear} implemented this strategy for a nonlinear Schrödinger equation and proved small-time approximate controllability between eigenmodes via bilinear interactions. 

This approach has two major consequences: First, it allows one to achieve \emph{global approximate controllability}. Second, because the saturation limit is performed on arbitrarily small time intervals, the method naturally leads to \emph{small-time controllability} results. In this sense, saturation provides a flexible and robust framework to ``amplify'' weak bilinear controls, turning them into a rich family of effective directions that ensure approximate controllability in large classes of nonlinear PDEs.
\medskip

\section{Preliminaries}
Throughout this article, we use the notation $X \lesssim Y$ if there exists a positive constant $C$, independent on the relevant parameters, such that $X \leq CY.$ Unless explicitly specified, function spaces are always defined as $\mathbb C$-vectorial spaces.
\subsection{Sobolev spaces and well-posedness}
We start by defining the functional framework of this work. For $s\in\R$, the Sobolev space $H^s(\T)$ is defined by
\[H^s(\T)
:=\big\{ u\in\mathcal{D}'(\T),\;
\|u\|_{H^s}^2
:=\sum_{k\in\Z}\langle k\rangle^{2s}
|\widehat u(k)|^2 <\infty \big\},\]
where $\widehat u(k)$ denotes the $k$--th Fourier coefficient of $u$ and $\langle k\rangle=(1+k^2)^{1/2}$. 
\begin{theorem}[Well-posedness of KdV-Sch in $H^{s}(\T)$]
Let $s\geq0$, $T>0$,  $\psi_0 \in H^{s}(\T)$, $u \in \PWC(0,T;\R^{q})$  and $Q\in C^{\infty}(\T;\R^{q})$. Then, there exists a unique  $\psi \in C([0,T];H^{s}(\T))$, solution of \eqref{eq:lkdv-sch}. Moreover, there exists $C=C(s,u,Q)>0$ such that
\[\|\psi(t)\|_{H^{s}} \le C e^{Ct}\|\psi_0\|_{H^{s}},\quad\forall t\in [0,T].\]
\end{theorem}
\begin{proof}
We denote by $L := -\partial_x^3 + i\alpha\partial_x^2$, $D(L)=H^{s+3}(\mathbb T)$, the skew-adjoint linear operator on $H^s(\mathbb T)$, which generates a unitary group $(e^{tL})_{t\in\mathbb R}$. Using Duhamel’s formula, we obtain
\begin{equation}\label{eq:duhamel}
\psi(t)
= e^{tL}\psi_0
+\int_0^t e^{(t-\tau)L}i(u\cdot Q)\psi(\tau)d\tau .
\end{equation}
Since $e^{tL}$ is unitary on $H^s(\mathbb T)$, we have $\|e^{tL}\psi_0\|_{H^s}=\|\psi_0\|_{H^s}$. Moreover, we also have
\[\|(u(t)\cdot Q)\psi\|_{H^s}\le C |u(t)|\|\psi\|_{H^s}.\]
Taking the $H^{s}$-norm in \eqref{eq:duhamel} and using the above estimates, we obtain
\[\|\psi(t)\|_{H^{s}}\le\|\psi_0\|_{H^{s}}+C\int_0^t |u(\tau)|\|\psi(\tau)\|_{H^{s}}d\tau.\]
Since $u\in \PWC(0,T;\R^{q})$, it follows that
\[\|\psi(t)\|_{H^{s}}
\le\|\psi_0\|_{H^{s}}+C\|u\|_{L^{\infty}(0,T)}
\int_0^t \|\psi(\tau)\|_{H^{s}}d\tau .\]
By Grönwall's lemma, there exists a constant $C:=C(s,u,Q)>0$ such that
\begin{equation}
\label{eq: Hs linear}\|\psi(t)\|_{H^{s}}\le C e^{Ct}\|\psi_0\|_{H^{s}},\quad \forall t\in[0,T].    
\end{equation}
Local existence and uniqueness follow from a fixed point argument applied to
the Duhamel formulation. The $H^s$ energy estimate \eqref{eq: Hs linear} prevents blow-up and allows us to extend the solution on $[0,T]$.
\end{proof}

\subsection{Saturating subspaces}

Let $n\in \N$, with $n\geq 2$. Let $\mathcal{H}$ be a finite-dimensional subspace of $C^\infty(\T;\R)$. Following the standard saturation approach, we define $\mathcal F_n(\mathcal H)$ as the largest vector subspace of $C^\infty(\T;\R)$ whose elements can be represented in the form
\[\phi_0-\sum_{j=1}^N\alpha_j(\phi_j')^n, \quad N\ge1,\quad \phi_0,\phi_1,\dots,\phi_N\in\mathcal H, \quad \alpha_j\geq 0,\]
we emphasize that the integer $N$ may depend on the element under consideration. When $n$ is odd, the coefficients $\alpha_j\ge0$ can be absorbed into the functions $\phi_j$, and in particular we have
\[\mathcal F_n(\mathcal H)=\mathcal H + \Span \big\{(\phi')^n,\; \phi\in\mathcal H\big\}.\] 
Note that $\mathcal F_n(\mathcal{H})$ is well-defined and finite-dimensional. 
\begin{definition}
We say that a finite-dimensional subspace $\mathcal{H}_0\subset C^\infty(\T,\R)$ is saturating if $\mathcal{H}_\infty$ is dense in $H^s(\T;\R)$,  for all $s\geq0$ where
\begin{equation*}\label{DefOfHInfty}
\mathcal H_\infty:=\bigcup_{j=0}^\infty \mathcal H_j,
\end{equation*}
with $\mathcal H_{j+1}:=\mathcal F_n(\mathcal H_j), \; j\ge0.$
\end{definition}
In the following result, we establish a saturation property for a general odd integer power $n\ge 3$. The main application in this article corresponds to the cubic case $n=3$, but the argument naturally extends to arbitrary odd $n\geq 3$, with no additional difficulty. For this reason, we present the construction in this broader framework, which may also be of independent interest.   We expect that an analogous statement should remain true for even powers, but our proof relies on the odd-power polarization argument and does not directly cover that case. 

\begin{proposition}\label{prop:saturating_general_power}
Let $n\in\N$ odd with $n\geq 3$. The finite dimensional space
\[\mathcal H_0=\Span\{1,\cos(x),\sin(x),\dots,\cos((n-1)x),\sin((n-1)x)\}\] is saturating.
\end{proposition}

\begin{proof}
We show that $\mathcal H_\infty$ contains all trigonometric modes.
Since $\cos(mx),\sin(mx)\in\mathcal H_0\subset\mathcal H_\infty$ for all $0\le m\le n-1$,
it is enough to prove that $\cos(Nx)$, $\sin(Nx)\in \mathcal H_\infty$ for every $N\ge n$. It is immediate that $\mathcal H_k\subset\mathcal H_{k+1}$ for any $k\in\N$.

\medskip

Note that, for any $f_1,\dots,f_n\in C^\infty(\T;\R)$,
\[
f_1' f_2'\cdots f_n' \in \mathcal F_n(\Span\{f_1,\dots,f_n\}).
\]
Indeed, for any $u_1,\dots,u_n\in\R$, one has the identity
\[
 u_1\cdots u_n
=\frac{1}{n!}
\sum_{\varepsilon\in\{0,1\}^n}
(-1)^{n-|\varepsilon|}
\big(\varepsilon_1u_1+\cdots+\varepsilon_nu_n\big)^n.
\]
Applying it pointwise with $u_j=f_j'(x)$ for any $x\in\T$, and using that
$\big(\sum \varepsilon_j f'_j\big){}^n\in\mathcal F_n(\mathcal H)$ yields the claim. Now assume that
\[\cos(mx),\sin(mx)\in\mathcal H_\infty \quad \text{for all } 0\le m\le N-1,\]
for some $N\geq n.$ Let $m_0=N-(n-1)$, so $1\le m_0\le N-1$.
By the induction hypothesis, $\sin(m_0x),\cos(m_0x)\in\mathcal H_\infty$ and $\sin(x)\in\mathcal H_0\subset\mathcal H_\infty$.
Choose $k\in\N$ large enough so that
\[
\sin(m_0x),\cos(m_0x),\sin(x) \in \mathcal H_k.
\]
Let $g(x)=\sin(x)$, so $g'(x)=\cos(x)$. By the previous analysis, we obtain $f'(g')^{n-1}\in\mathcal H_{k+1}$ for every $f\in\mathcal H_k$.
By setting $f_1(x)=\tfrac{1}{m_0}\sin(m_0 x)$ and $f_2(x)=-\tfrac{1}{m_0}\cos(m_0x)$, we obtain
\[ f_1'(x)=\cos(m_0x),
\quad   f_2'(x)=\sin(m_0x).
\]
As a result,
\[f_1'(g')^{n-1}=\cos(m_0x)\cos^{n-1}(x) \in \mathcal H_{k+1}\subset\mathcal H_\infty,
\quad f_2'(g')^{n-1}=\sin(m_0x)\cos^{n-1}(x) \in \mathcal H_{k+1}\subset\mathcal H_\infty.\]
 Recall that for any $q\in\N$, we have the standard expansion,
\[
\cos^q (x) = 2^{-q}\sum_{r=0}^q \binom{q}{r}\cos((q-2r)x).
\]
With $q=n-1$, multiplying by $\cos(m_0x)$ and using standard trigonometric identities, we obtain
\[
\cos(m_0x)\cos^{n-1}x
=
2^{-n}\cos((m_0+n-1)x) + \sum_{\ell\le m_0+n-3} a_\ell \cos(\ell x),
\]
for some coefficients $a_\ell\in\R$.
Since $m_0+n-1=N$ and $m_0+n-3=N-2$, all modes in the sum have frequency at most $N-2$.
By the induction hypothesis, each $\cos(\ell x)$ with $\ell\le N-2$ belongs to $\mathcal H_\infty$.
Hence, $\cos(Nx)\in\mathcal H_\infty$.

\noindent Similarly,
\[\sin(m_0x)\cos^{n-1}x=2^{-n}\sin((m_0+n-1)x) + \sum_{\ell\le m_0+n-3} b_\ell \sin(\ell x),\]
and we also obtain $\sin(Nx)\in\mathcal H_\infty$. This proves the induction step, and thus $\sin(nx),\cos(nx)\in\mathcal H_\infty$ for any $n\ge0$.
\end{proof}
\section{Bilinear control of phases}
Throughout this section, for any initial datum $\psi_{0}$ and any control
$u(\cdot)$, we denote by $R_{t}(\psi_{0},u):=\psi(t,\psi_0,u)$
the value at time $t$ of the solution of \eqref{eq:lkdv-sch}.

\subsection{Saturation limit}
Let $\varphi \in C^{3}(\T;\R)$, $u\in\R^q$ and $\tau>0$. We define:
\[
H=i(\varphi^{\prime})^3+i(u\cdot Q).\]
\begin{proposition}
\label{prop: sat}
Let $\psi_0 \in L^2(\T)$, $\|\psi_0\|_{L^{2}(\T)}=1$, $u \in \R^q$, and 
$\varphi \in C^3(\T;\R)$. Then,
\[
e^{ i\tau^{-\frac{1}{3}}\varphi}R_{\tau}(e^{- i\tau^{-\frac{1}{3}}\varphi}\psi_0, \tau^{-1}u\big)
\longrightarrow
e^{H}\psi_0
\quad\text{in }  L^2(\T) \text{ as } \tau \to 0^+.\]
\end{proposition}
\begin{proof}
To justify the integrations by parts, we first assume that $\psi_0\in H^3(\T)$. In this case, the corresponding solution is sufficiently regular and all the following computations are rigorous. The general case $\psi_0\in L^2(\T)$ then follows by a standard density argument. For $\psi_0\in H^3(\T),$ we moreover consider a regularized version of $\psi_0$, given by an appropriate scaling of the action of heat semigroup
\begin{equation*}
\psi_{0}^{\tau} := e^{\tau^{1/4}\partial_x^2}\psi_0,\quad \tau>0.
\end{equation*}
 Using regularizing properties of the heat equation, for any $m\in\R$ and $s\leq 3$, there exists $C_{s,m}>0$ such that
\begin{equation}
\label{eq: w reg est}
\|\psi_{0}^{\tau}\|_{H^{s+m}(\T)}  \leq C_{s,m} \tau^{-m/8}\|\psi_0\|_{H^s(\T)}.\end{equation}
In addition,
\[\psi_{0}^{\tau} \longrightarrow \psi_0 \quad \text{in } L^2(\T), \quad \text{as } \tau \to 0^+.\] 
We set, for $\tau>0$,
\begin{equation*}
\phi(t)=e^{ i\tau^{-\frac{1}{3}}}R_{t}(e^{ -i\tau^{-\frac{1}{3}}}\psi_0, \tau^{-1}u\big), \quad w(t)=e^{Ht}\psi_{0}^{\tau},
\end{equation*}
and $v(t)=\phi(\tau t)-w(t)$. In this setting, proving \Cref{prop: sat} reduces $\|v(1)\|_{L^{2}(\T)} \to 0$ as $\tau \to 0^+$. Then, $v$ satisfies

\[
\begin{cases}
\partial_t v-\tau e^{ i\tau^{-\frac{1}{3}}\varphi}Le^{- i\tau^{-\frac{1}{3}}\varphi}v-i(u\cdot Q) v =\mathcal{F}_\tau w& \text{ in }\R_+\times\T\\
 v(t=0) = \psi_0-\psi_0^\tau & \text{ in }\T,
 \end{cases}  
 \]
 where 
 \begin{align*}
 \mathcal F_\tau w :=& e^{ i\tau^{-\frac{1}{3}}\varphi}Le^{- i\tau^{-\frac{1}{3}}\varphi}w+i(u\cdot Q)w-Hw  \\
 =&-\tau\partial_x^3w+a(x)\partial_x^2w+b(x)\partial_xw+c(x)w,
 \end{align*}
 with
$$
a=3i\tau^{\frac{2}{3}}\varphi^{\prime}+i\tau\alpha,\quad  b=3i\tau^{\frac{2}{3}}\varphi^{\prime\prime}+3\tau^{\frac{1}{3}}(\varphi^{\prime})^2+2\tau^{\frac{2}{3}}\varphi^{\prime},$$
$$c=i\tau^{\frac{2}{3}}\varphi^{\prime\prime\prime}+3\tau^{\frac{1}{3}}\varphi^{\prime}\varphi^{\prime\prime}-i\tau^{\frac{1}{3}}(\varphi^{\prime})^2 + \tau^{\frac{2}{3}}\varphi^{\prime\prime}.$$
As $e^{ i\tau^{-\frac{1}{3}}}Le^{ -i\tau^{-\frac{1}{3}}}$ is skew-adjoint, we have
\[\tfrac12\tfrac{d}{dt}\|v(t)\|_{L^2}^2
=\Re\langle F_\tau w(t),v(t)\rangle.\]
By recalling \eqref{eq: w reg est}, 
\[\|\partial_x^{k} w(t)\|_{L^{2}} \lesssim\|\psi_{0}^{\tau}\|_{H^{k}}
\lesssim\tau^{-k/8}\|\psi_{0}\|_{L^{2}}, \quad k\in\{0,1,2,3\}.\]
Therefore,
\[\big|\Re\int_{\T} \tau\partial_x^{3} w \;dx
\big|
\lesssim \tau^{5/8} \|\psi_{0}\|_{L^{2}}\|v\|_{L^{2}}.\]
With the same argument, using $\|a\|_{\infty}=O(\tau^{2/3})$, $\|b\|_{\infty},\|c\|_{\infty}=O(\tau^{1/3})$, we obtain for $\tau>0$ sufficiently small
\[\big|\Re\int_{\T} \mathcal F_{\tau}w\overline v \; dx
\big|\lesssim\tau^{5/24}\|\psi_{0}\|_{L^{2}} \|v\|_{L^{2}}.\]
 Thus $\tfrac{d}{dt}\|v(t)\|_{L^{2}}^{2}
\lesssim
\tau^{5/24}\|\psi_{0}\|_{L^{2}}\|v(t)\|_{L^{2}}$, which yields
\[\|v(1)\|_{L^{2}}
\le\|v(0)\|_{L^{2}} + C\tau^{5/24}\|\psi_{0}\|_{L^{2}}.\]
Since $v(0)=\psi_{0}-\psi_{0}^{\tau}\to0$ in $L^{2}$ as $\tau\to0$, we conclude
\[
\|v(1)\|_{L^{2}}\longrightarrow 0\text{ as }\tau\to0^{+}.\]
\end{proof}
\begin{definition}[Admissible phases]\label{AdmissiblePhasesDef}
 We define by $\mathcal A$ the set of phases $\theta \in L^{2}(\T;\R)$ such that,
for every $\psi_{0}\in L^{2}(\T)$, $T>0$ and $\varepsilon>0$,
there exist $\tau\in(0,T)$ and a control $u\in \PWC(0,\tau;\R^q)$ such that
\begin{equation}
\label{eq:defA}
\|R_{\tau}(\psi_{0},u) - e^{i\theta}\psi_{0}\|_{L^{2}}\le
\varepsilon.
\end{equation}   
\end{definition}

From \cref{prop: sat} with $\varphi=0$, we immediately obtain
$\operatorname{span}\{Q_{0},\dots,Q_{q-1}\}\subset \mathcal A$. Moreover, if
$\theta_{1},\theta_{2}\in\mathcal A$ and $c\in\R$, then
$c\theta_1$ and $\theta_{1}+\theta_{2}$ belong to $\mathcal A$, by re-scaling the control
and by concatenating the corresponding control laws. More precisely, given
$u_1:[0,T_1]\to\R^q$ and $u_2:[0,T_2]\to\R^q$, we define the concatenated control
$u_1\diamond u_2:[0,T_1+T_2]\to\R^q$ by
\[
(u_1\diamond u_2)(t)=
\begin{cases}
u_1(t) & t\in[0,T_1],\\
u_2(t-T_1) & t\in(T_1,T_1+T_2].
\end{cases}
\]

Assume that the solution map $R_t(\psi_0,u)$ is well-defined on $[0,T]$ for $u=u_1\diamond u_2$,
for some $T\in(T_1,T_1+T_2]$. Then, for every $t\in(0,T-T_1)$, one has the flow identity
\[R_{T_1+t}(\psi_0,u_1\diamond u_2)=R_t\left(R_{T_1}(\psi_0,u_1),u_2\right).\]

\begin{proposition}
\label{prop:A-cubic}
Let $\varphi \in C^{3}(\T;\R)$ and assume that 
$\operatorname{span}\{\varphi\}\subset \mathcal A$. 
Then the phase $(\varphi')^{3}$ also belongs to $\mathcal A$.
\end{proposition}

\begin{proof}
Let $\psi_{0}\in L^{2}(\T)$, $T>0$, and $\varepsilon\in(0,1)$.  
By \cref{prop: sat}, there exists $\tau_{1}\in(0,T/3)$ such that
\begin{equation*}
\label{eq:step1}
\bigl\|e^{i\varphi\tau_{1}^{-1/3}}
R_{\tau_{1}}\bigl(e^{-i\varphi\tau_{1}^{-1/3}}\psi_{0},0\bigr)-e^{i(\varphi')^{3}}\psi_{0}
\bigr\|_{L^{2}}\le\frac{\varepsilon}{3}.
\end{equation*}
Set $\psi_{1}:=R_{\tau_{1}}\bigl(e^{-i\varphi\tau_{1}^{-1/3}}\psi_{0},0\bigr)\in L^{2}(\T)$, so that
\begin{equation}
\label{eq:step2}
\bigl\|e^{i\varphi\tau_{1}^{-1/3}}\psi_{1}-e^{i(\varphi')^{3}}\psi_{0}
\bigr\|_{L^{2}}\le\frac{\varepsilon}{3}.
\end{equation}
As $\varphi\tau_{1}^{-1/3}\in\mathcal A$, there exist $\tau_{2}\in(0,T/3)$ and a control
$u_{2}\in \PWC(0,\tau_{2};\R^q)$ such that
\begin{equation}
\label{eq:step3}
\bigl\|R_{\tau_{2}}(\psi_{1},u_{2})-e^{i\varphi\tau_{1}^{-1/3}}\psi_{1}\bigr\|_{L^{2}}
\le\frac{\varepsilon}{3}.
\end{equation}
Moreover, as $-\varphi\tau_{1}^{-1/3}\in\mathcal A$, there exist $\tau_{0}\in(0,T/3)$ and a control
$u_{0}\in \PWC(0,\tau_{0};\R^q)$ such that
\begin{equation}
\label{eq:step4}
\bigl\|R_{\tau_{0}}(\psi_{0},u_{0})-e^{-i\varphi\tau_{1}^{-1/3}}\psi_{0}\bigr\|_{L^{2}}
\le \frac{\varepsilon}{3}.
\end{equation}
By the $L^2$-isometry of the free flow $R_t(\cdot,0)=e^{tL}$ and  by \eqref{eq:step4}, we obtain
\begin{equation}\label{eq:step5}\begin{aligned}
\bigl\|R_{\tau_{0}+\tau_{1}}(\psi_{0},u_{0}\diamond\mathbf 0)-\psi_{1}\bigr\|_{L^{2}}&=\bigl\|e^{\tau_1 L}(R_{\tau_{0}}(\psi_{0},u_{0})-e^{i\varphi\tau_{1}^{-1/3}}\psi_{0})\bigr\|_{L^{2}}\\
&=\bigl\|R_{\tau_{0}}(\psi_{0},u_{0})-e^{i\varphi\tau_{1}^{-1/3}}\psi_{0}\bigr\|_{L^{2}}\\
&\leq \frac{\varepsilon}{3},
\end{aligned}\end{equation}
where $\mathbf 0$ denotes the identically zero control on $(0,\tau_1)$. Define the concatenated control $u:=u_0\diamond(\mathbf 0\diamond u_2)$.  
Using \eqref{eq:step3} and \eqref{eq:step5}, we obtain
\[\bigl\|R_{\tau_0+\tau_1+\tau_2}(\psi_0,u)-e^{i\varphi\tau_{1}^{-1/3}}\psi_{1}\bigr\|_{L^{2}}
\le\frac{2\varepsilon}{3},\]
combining with \eqref{eq:step2} yields
\[\bigl\|R_{\tau_{0}+\tau_{1}+\tau_{2}}(\psi_{0},u)-e^{i(\varphi')^{3}}\psi_{0}\bigr\|_{L^{2}}
\le\varepsilon.\]
Since $\tau_{0}+\tau_{1}+\tau_{2}<T$, this proves that $(\varphi')^{3}\in\mathcal A$.
\end{proof}

Now, we state that $H^s(\T;\R)\subset\mathcal{A}$, for all $s\geq0$, which can be rewritten as the following \textit{small-time approximate controllability of phases.}
\begin{theorem}[Small-time $L^2$ control of phases]
\label{th: phases}
For every $\psi_0 \in L^2(\T)$, $\eps>0$ and every phase 
$\theta\in C^{\infty}(\T;\R)$, there exists a time $T\in(0,\eps)$ and a control $u\in \PWC(0,T;\R^q)$ such that the corresponding solution of \eqref{eq:lkdv-sch} satisfies
\[\|\psi(T)-e^{i\theta}\psi_{0}\|_{L^{2}} < \varepsilon .\]
\end{theorem}
\begin{proof}
By \Cref{prop:A-cubic}, the vectorial space generated by iterating the operation $\varphi\mapsto(\varphi')^{3}$ starting from 
$\operatorname{span}\{Q_{0},\dots,Q_{q-1}\}$ is contained in $\mathcal A$.  
By assumption \eqref{eq: controls} and \Cref{prop:saturating_general_power}, this set is dense in $H^{s}(\T)$. Since the mapping $L^{2}(\T)\ni \phi \longmapsto e^{i\phi}\psi_{0}\in L^{2}(\T)$ is continuous,  we can choose $\phi_{\varepsilon}\in\mathcal A$ such that
\[\| e^{i\theta}\psi_{0} - e^{i\phi_{\varepsilon}}\psi_{0}\|_{L^{2}}
\le\frac{\varepsilon}{2}.\]
As $\phi_{\varepsilon}\in\mathcal A$, there exist
$T\in(0,\eps)$ and $u\in \PWC(0,T;\R^q)$ (see \cref{AdmissiblePhasesDef}) such that
\[\| R_{T}(\psi_{0},u) - e^{i\phi_{\varepsilon}}\psi_{0}\|_{L^{2}}
\le\frac{\varepsilon}{2}.\]
Combining the two estimates yields $\| R_{T}(\psi_{0},u) - e^{i\theta}\psi_{0}\|_{L^{2}}< \varepsilon$, which completes the proof. 
\end{proof}
\section{Control of flows}
We now combine the small-time control of phases with the free dynamics generated by
\[L=-\partial_x^3+i\alpha\partial_x^2\]
in order to produce, in small time, unitary operators associated with transport flows on $\T$.
This part is strongly inspired by the techniques developed in \cite{BeauchardPozzoli2025AIHPC}.

\begin{definition}A unitary operator $S$ on $L^2(\T)$ is $L^2$--STAR if, for every $\psi_0 \in L^{2}(\T)$ with $\|\psi_0\|_{L^{2}(\T)}=1$ and $\eps>0$, there exists $T\in[0,\eps]$, $\beta \in [0,2\pi)$ and $u \in \PWC(0,T;\R^q)$ such that
\[\|R_{T}(\psi_0,u)-e^{i\beta}S\psi_0\|_{L^{2}(\T)}<\eps.\]
\end{definition}
\begin{definition}
 We denote by $\mathrm{Diff}^0(\mathbb T)$ the group of $C^1$ diffeomorphisms of
$\mathbb T$ isotopic to the identity (equivalently, orientation-preserving $C^1$ diffeomorphisms). For $P\in \mathrm{Diff}^0(\mathbb T)$, set $J_P(x):=P'(x)>0$ and define
\[(U_P\psi)(x):=\sqrt{J_P(x)}\psi(P(x))=\sqrt{P'(x)}\psi(P(x)),\quad \psi\in L^2(\mathbb T;\mathbb C).\]

We also denote by  $\mathrm{Vec}(\T)$ the space of smooth real
vector fields on $\T$.
\end{definition}

For every $P\in \mathrm{Diff}^0(\mathbb T)$, the operator $U_P$ is unitary on $L^2(\mathbb T)$, in particular $\|U_P\psi\|_{L^2}=\|\psi\|_{L^2}$. Each $X_f\in \mathrm{Vec}(\T)$ can be identified by its coefficient $f\in W^{1,\infty}(\mathbb T;\mathbb R)$. Therefore, we use this identification along this part. For $f\in W^{1,\infty}(\mathbb T;\mathbb R)$, let $\phi_f^t\in\mathrm{Diff}^0(\mathbb T)$ be defined by
\[
\begin{cases}
\partial_t\phi_f^t(x)=f(\phi_f^t(x))&(t,x)\in\R\times \T,\\
\phi_f^0(x)=x&x\in \T.
\end{cases}
\]
Define the transport operator
\[T_f:=f\partial_x+\frac12 f',
\quad D(T_f):=\big\{\psi\in L^2(\mathbb T),\;  f\partial_x\psi\in L^2(\mathbb T)\big\}.
\]
Then $T_f$ generates a unitary group $(e^{tT_f})_{t\in\mathbb R}$ on $L^2(\mathbb T;\mathbb C)$ and $e^{tT_f}=U_{\phi_f^t}$, for $t\in\mathbb R.$ In particular, for every $t\in\mathbb R$ and $\psi\in L^2(\mathbb T)$,
\[
(e^{tT_f}\psi)(x)
=\psi(\phi_f^t(x))\exp\!\Big(\tfrac12\int_0^t f'(\phi_f^s(x))\,ds\Big)
=\sqrt{\partial_x\phi_f^t(x)}\,\psi(\phi_f^t(x)).
\]

\subsection{A small-time generation of transports}
In the previous section, we proved that for every phase $\theta \in L^{2}(\T;\R)$, the unitary multiplication operator \[M_\theta: \psi \mapsto e^{i\theta}\psi\] 
is $L^2$--STAR. We combine controllability for phases with the free group $e^{tL}$, to produce in small-time unitary operators associated with flows/transport. For $\tau>0$, $n\in \N$, $\varphi,g_\tau\in C^{\infty}(\T;\R)$, we define

\begin{equation}
\label{eq: Wtaun}
W_{\tau,n}=\big(e^{\frac{ig_\tau}{n\sqrt{\tau}}}e^{\frac{i\varphi}{\sqrt{\tau}}}e^{\frac{\tau L}{n}}e^{-\frac{i\varphi}{\sqrt{\tau}}}\big)^{n}   .
\end{equation}
Observe that $W_{\tau,n}$ is a composition of unitary operators: phase shifts (which are $L^2$--STAR by the previous section) and  the uncontrolled dynamic $e^{\frac{\tau L}{n}}$ in short time. In this direction, the following lemma ensures reachability properties of the compositions and limit of $L^2$--STAR operators.
\begin{lemma}[Lemma 6, \cite{BeauchardPozzoli2025AIHPC}]
\label{lem:closure-L2STAR}
The class of $L^2$--\emph{STAR} operators is stable under composition and under strong operator limits.
More precisely:
\begin{enumerate}
\item If $S_1$ and $S_2$ are $L^2$--\emph{STAR}, then $S_2S_1$ is $L^2$--\emph{STAR}.
\item If $(S_n)_{n\ge1}$ is a sequence of $L^2$--\emph{STAR} unitary operators and $S_n \to S$
strongly on $L^2(\T)$, then $S$ is $L^2$--\emph{STAR}.
\end{enumerate}
\end{lemma}
We also recall a specific conjugation identity and the Trotter--Kato product formula.
\begin{proposition}[Proposition 11, \cite{beauchard2025examples}]
\label{prop: conj}
Let $A,B$ be self-adjoint operators on a Hilbert space $H$, and suppose that $e^{iB}$ is an isomorphism of $\mathcal D$, where $\mathcal D$ is a core for $A$. Then, for any $t\in\mathbb R$,
\[ e^{iB}e^{itA}e^{-iB}=\exp\bigl(e^{iB}(itA)e^{-iB}\bigr)
=\exp\bigl(ite^{iB}Ae^{-iB}\bigr).\]
\end{proposition}
\begin{proposition}[Theorem~VIII.31, \cite{reed2012methods}]
\label{prop: TK}
Let $A,B$ be self-adjoint operators on a Hilbert space $H$, with domains
$\mathcal D(A)$ and $\mathcal D(B)$, and assume that $A+B$ is essentially
self-adjoint on $\mathcal D(A)\cap \mathcal D(B)$. Then, for every $\psi_0\in H$
and every $t\in\mathbb R$,
\[
\bigl\|\bigl(e^{itA/n}e^{itB/n}\bigr)^n \psi_0-e^{it(A+B)}\psi_0\bigr\|_{H}\xrightarrow[n\to\infty]{}0.
\]
\end{proposition}
\begin{proposition}\label{PropTransportSTAR}
Let $\varphi \in C^\infty(\T;\R)$ and, for $\tau \ge 0$, set $g_\tau=(\varphi')^3+\alpha\sqrt{\tau}\,(\varphi')^2$. For every $\psi \in L^{2}(\T)$, we have
\[\lim_{\tau\to0}\Bigl(\lim_{n\to\infty}W_{\tau,n}\psi\Bigr)=e^{T_f}\psi,
\quad f=3(\varphi')^2.\]
\end{proposition}
\begin{proof}
Using \Cref{prop: conj} in \eqref{eq: Wtaun}, we easily get
\begin{equation*}
\label{eq: Wtaun2}
W_{\tau,n}=\big(e^{\frac{ig_\tau}{n\sqrt{\tau}}}\exp{\big(\tfrac{\tau}{n}e^{\frac{i\varphi}{\sqrt{\tau}}} Le^{-\frac{i\varphi}{\sqrt{\tau}}}\big)}\big)^{n}    ,
\end{equation*}
and applying \Cref{prop: TK}, we deduce that $W_{\tau,n} \longrightarrow W_{\tau}$ strongly as $n \to \infty$, where
\begin{equation*}
\label{eq: Wtau}
W_{\tau}=\exp{\big(\tfrac{ig_\tau}{\sqrt{\tau}}+\tau e^{\frac{i\varphi}{\sqrt{\tau}}} Le^{-\frac{i\varphi}{\sqrt{\tau}}}\big)}.    
\end{equation*}
Explicit computations yields
\[\begin{aligned}
  \frac{ig_\tau}{\sqrt{\tau}}+\tau e^{\frac{i\varphi}{\sqrt{\tau}}} Le^{-\frac{i\varphi}{\sqrt{\tau}}}=&\tau L +T_{f}+\sqrt{\tau}\big(3i\varphi^{\prime}\partial_x^2+(3i\varphi^{\prime\prime}+2\alpha\varphi^{\prime})\partial_x+i\varphi^{\prime\prime\prime}+\alpha\varphi^{\prime\prime}\big)\\
  &+\frac{ig_\tau}{\sqrt{\tau}}-i\alpha(\varphi^{\prime})^{2}-\frac{i(\varphi^{\prime})^{3}}{\sqrt{\tau}},  
\end{aligned}\]
where $f=3(\varphi^{\prime})^2\in C^\infty(\T;\R)$. As $g_\tau=(\varphi^{\prime})^{3}+\alpha\sqrt{\tau}(\varphi^{\prime})^{2}$, we observe that
\[W_\tau:=\exp{\left(\tau L + T_f + \sqrt{\tau}K\right)}, \quad f=3(\varphi^{\prime})^2,\]
where $K$ is a differential operator of order at most $2$ with smooth coefficients depending on $\varphi$ and $\alpha$. Define $A_\tau=\tau L + T_f + \sqrt{\tau}K$  and $A_0=T_f$, so that $W_\tau=e^{A_\tau}$ and $e^{T_f}=e^{A_0}$. Let $\psi \in H^{3}(\T)$. By the Duhamel formula
\[e^{A_\tau}\psi-e^{A_0}\psi=\int_{0}^{1}e^{(1-s)A_\tau}(A_\tau - A_0)e^{sA_0}\psi ds.\]
Yet, $A_\tau - A_0=\tau L+\sqrt{\tau}K: H^3(\T) \rightarrow L^{2}(\T)$ and $(e^{sA_0})_{s\in[0,1]}=(e^{sT_f})_{s\in[0,1]}$ is uniformly bounded on
$H^3(\T)$. Indeed, if $z(s)=e^{sT_f}\psi$ then 
\[\partial_s\partial_x^k z=T_f(\partial_x^k z)+[\partial_x^k,T_f]u,\quad k\le 3,\]
as $T_f$ is skew-adjoint, we have
\[\|e^{sT_f}\psi\|_{H^3(\T)}\le C_f\|\psi\|_{H^3(\T)},
\quad s\in[0,1].\]
Therefore,
\[\|e^{A_\tau}\psi-e^{A_0}\psi\|_{L^{2}(\T)}\leq C\int_{0}^{1}\|(A_\tau - A_0)e^{sA_0}\psi\|_{L^{2}(\T)}ds\]
and
\[\|W_\tau\psi - e^{T_f}\psi\|_{L^{2}(\T)}=\|e^{A_\tau}\psi-e^{A_0}\psi\|_{L^{2}(\T)}\leq C(\tau+\sqrt{\tau})\|\psi\|_{H^{3}(\T)}.\]
By density $W_\tau \longrightarrow e^{T_f}$ strongly and $\tau \to 0$.
\end{proof}

We introduce the following class of profiles 
\[G=\{(\varphi^{\prime})^2,\; \varphi\in C^{\infty}(\T;\R)\}.\]
Note that Proposition \ref{PropTransportSTAR} states that for every $f\in G$, the operator $e^{T_f}$ is $L^2$--STAR. The prefactor $3$ plays no structural role since $G$ is stable under positive scaling. Indeed, for any $\lambda\in \R$,  $|\lambda|(\varphi')^2 = (\sqrt{|\lambda|}\,\varphi')^2$) and $T_{\lambda f}=\lambda T_f$. The main objective of the next section is to prove the following result.
\begin{proposition}
\label{prop: spanG}
Let $f \in \mathrm{span}(G)$. Then $e^{T_f}$ is $L^2$--STAR.
\end{proposition}
\subsection{A cone property}
We denote by $\mathrm{cone}(G)$ the convex cone generated by $G$, i.e,
\[\mathrm{cone}(G)=\big\{\sum_{j=1}^m \alpha_j g_j,\;m\in \N^*, \; g_j \in G, \; \alpha_j\ge0\big\}.\]
\begin{proposition}
\label{prop: coneG}
For every $f\in \mathrm{cone}(G)$, $e^{T_f}$ is $L^2$--STAR. 
\end{proposition}
\begin{proof}
As $G$ is stable under multiplication by positive constants, it is enough to prove
that for every $f_1,f_2\in G$, the operator $e^{T_{f_1+f_2}}$ is $L^2$--STAR. Let $f_1,f_2\in G$ and $n\in\mathbb N$. By scaling with $\lambda=1/n$,
\[
\tfrac{1}{n}f_j \in G \quad\Rightarrow\quad S_{j,n}=e^{T_{\frac{1}{n} f_j}}=e^{\frac{1}{n}T_{f_j}}
\ \text{is $L^2$--STAR}\quad (j=1,2).\]
By stability under composition, we have that $S_{1,n}S_{2,n}$ is $L^2$--STAR and 
\[U_n:=\bigl(S_{1,n}S_{2,n}\bigr)^n=\bigl(e^{\frac{1}{n}T_{f_1}}e^{\frac{1}{n}T_{f_2}}\bigr)^n.\]
Finally, \Cref{prop: TK} infers that $e^{T_{f_1}+T_{f_2}}=e^{T_{f_1+f_2}}$ is $L^2$--STAR.    
\end{proof}
\begin{corollary}\label{cor: translations}
For every $\delta>0$, the translation group $e^{\delta\partial_x}$ is $L^2$--STAR.
\end{corollary}
\begin{proof}Consider $\varphi_1=\cos(x)$ and $\varphi_2=\sin(x)$,  and set $f_1=(\varphi_1')^2=\sin^2 (x)$,
$f_2=(\varphi_2')^2=\cos^2 (x).$ Then for any $\delta>0$ and any $f=f_1+f_2$, $e^{\delta T_{f}}$ is $L^2$--STAR. 
\end{proof}
In \cite{BeauchardPozzoli2025AIHPC}, the admissible set of drift fields is $\{\varphi^{\prime}: \varphi \in C^{\infty}(\T;\R)\}$. In our framework, however, we can only generate transport operators of the form $e^{T_f}$, with $f=(\varphi^{\prime})^{2}\geq 0$, therefore at first sight the sign constraint might look like an obstruction to producing negative transport. Nevertheless, in the periodic case, we can still obtain translation in both directions. 

\begin{proposition}\label{prop: neg_cone}
Let $f \in \mathrm{cone}(G)$. Then $e^{-T_f}=e^{T_{-f}}$ is $L^2$--\emph{STAR}.
\end{proposition}

\begin{proof}
\textbf{Step 1.} Take $f \in \mathrm{cone}(G)$ and $\delta>0$, and set $g:=f+\delta$.
Since $1\in\mathrm{cone}(G)$, we have $g\in\mathrm{cone}(G)$ and
$g\ge \delta>0$. Let
\[\Pi(g):=\int_{0}^{2\pi}\frac{1}{g(x)}\,dx\]
be the period of the flow associated with $g$. Thus for all $n\in \N^{*}$,
\[\phi_g^{-\lambda}=\phi_g^{n\Pi(g)-\lambda}, \quad \forall \lambda \in [0,\Pi(g)].\]
For any $\kappa>0$, there exists $k\in \N$ such that
\[k\Pi(g)\leq \kappa < (k+1)\Pi(g),
\]
Let $s:=(k+1)\Pi(g)-\kappa>0.$ Then $\phi_g^{-\kappa}=\phi_g^{s}$ and therefore
\[U_{\phi_{g}^{-\kappa}}=e^{-\kappa T_{g}}=U_{\phi_g^{s}}=e^{sT_{g}}=e^{T_{sg}}.\]
Since $sg\in\mathrm{cone}(G)$, \Cref{prop: coneG} yields that $e^{-\kappa T_{g}}$ is $L^2$--\emph{STAR}.
Thus, for any $\eps,\kappa>0$ and any $\psi_{0} \in L^{2}(\T)$ with $\|\psi_{0}\|_{L^{2}(\T)}=1$,
there exist $T\in[0,\eps]$, $\beta\in [0,2\pi)$ and $u\in \PWC(0,T;\R^q)$ such that
\[\|\psi(T,u,\psi_0)-e^{i\beta}e^{-\kappa T_{g}}\psi_0\|_{L^{2}(\T)}
=
\|\psi(T,u,\psi_0)-e^{i\beta}e^{sT_{g}}\psi_0\|_{L^{2}(\T)}
<\eps.\]
\textbf{Step 2.} Since $-f=-g+\delta$, we have
\[T_{-f}=T_{-g}+T_\delta=T_{-g}+\delta\partial_x.\]
By Step~1 (applied with $\kappa=\frac1n$), $e^{\frac1nT_{-g}}=e^{-\frac1nT_g}$ is $L^2$--\emph{STAR}
for every $n\in\N^*$. Moreover, by \Cref{cor: translations}, $e^{\frac{\delta}{n}\partial_x}$ is
$L^2$--\emph{STAR} for every $n\in\N^*$. Therefore, by \Cref{prop: TK},
\[e^{-T_f}=e^{T_{-f}}=e^{T_{-g}+\delta\partial_x}
=\lim_{n\to \infty}\Bigl(e^{\frac{1}{n}T_{-g}}\,e^{\frac{\delta}{n}\partial_x}\Bigr)^n
\quad\text{strongly in }L^2(\T).\]
We conclude that $e^{T_{-f}}=e^{-T_f}$ is $L^2$--\emph{STAR}.
\end{proof}

Finally \Cref{prop: spanG} is a direct consequence of the previous result.

\subsection{Lie properties}
We define
\[\mathcal L:=\bigl\{f\in \mathrm{Vec}(\T),\;\forall t\in\R,\ U_{\phi_f^{t}}\ \text{is }L^{2}\text{--STAR}\bigr\}.\]
Since $e^{T_h}$ is $L^2$--\emph{STAR} for every $h\in \mathrm{span}(G)$ and
$e^{tT_h}=e^{T_{th}}$, we indeed have $\mathrm{span}(G)\subset \mathcal L$.

\begin{proposition}[{\cite[Theorem~15]{BeauchardPozzoli2025AIHPC}}]\label{prop: L_Lie}
The set $\mathcal L$ is a Lie subalgebra of $\mathrm{Vec}(\T)$.
\end{proposition}

Let
\[\mathcal G:=\mathrm{Lie}\bigl(\{f\partial_x,\;f\in G\}\bigr)\subset \mathrm{Vec}(\T).\]
Recall that in dimension one we have
$[f\partial_x,g\partial_x]=(f g'-g f')\,\partial_x$.

\begin{proposition}\label{prop: LieG}
Let $p$ a trigonometric polynomial. Then $p\partial_x \in \mathcal G$.
\end{proposition}

\begin{proof}
\textbf{Step 1.} Let $\varphi_1(x)=\cos(x)$ and $\varphi_2(x)=\sin(x)$. Then for
\[
f_1:=(\varphi_1')^2=\sin^2 x,\quad f_2:=(\varphi_2')^2=\cos^2 x,
\]
we have $f_1\partial_x,f_2\partial_x\in\mathcal G$. Since $\mathcal G$ is a vector space,
\[
(f_1+f_2)\partial_x=\partial_x\in\mathcal G,
\quad
(f_2-f_1)\partial_x=\cos(2x)\,\partial_x\in\mathcal G.
\]
Moreover, $[f_1\partial_x,f_2\partial_x]=(f_1 f_2'-f_2 f_1')\partial_x.$ Using $f_1'=\sin(2x)$ and $f_2'=-\sin(2x)$, we obtain
\[
f_1 f_2'-f_2 f_1'
=
-\sin^2(x)\sin(2x)-\cos^2(x)\sin(2x)
=-\sin(2x).
\]
Thus $\sin(2x)\partial_x\in\mathcal G$. Summing up, we obtain
\[\{1,\cos(2x),\sin(2x)\}\partial_x\subset \mathcal G.\]
\textbf{Step 2.} Set $\varphi_\pm(x):=\sin(x)\pm \frac12\sin(2x)$, so that
\[
\varphi_\pm'(x)=\cos(x)\pm \cos(2x),
\quad
f_\pm:=(\varphi_\pm')^2=(\cos(x)\pm \cos(2x))^2.
\]
and $f_\pm\partial_x\in\mathcal G$. Define
\[
g:=\frac12(f_+-f_-)=2\cos(x)\cos(2x)=\cos(x)+\cos(3x),
\]
so $g\partial_x\in\mathcal G$. Since $\partial_x\in\mathcal G$, we also have
\[[\partial_x,g\partial_x]=g'\partial_x=(-\sin(x)-3\sin(3x))\partial_x\in\mathcal G.\]
Similarly $g^{(k)}\partial_x \in \mathcal G$ for all $k \in \N$. Finally,
\[\cos(x)=\tfrac{1}{8}(9g+g^{\prime\prime}),\quad
    \sin(x)=-\tfrac{1}{8}(9g^{\prime}+g^{\prime\prime\prime})\]
This step proves 
\[\{1,\sin(x),\cos(x), \cos(2x),\sin(2x)\}\partial_x\subset\mathcal G.\] 

\noindent \textbf{Step 3.}
We prove by induction that for every $n\ge1$,
\[\{\sin(nx),\cos(nx)\}\partial_x\subset\mathcal G.\] 
The cases $n=1,2$ follow from Step~2. Now, assume that for any $n\ge2$ one has
\[\{\sin(nx),\cos(nx),\sin((n-1)x),\cos((n-1)x)\}\partial_x\subset\mathcal G.\]
Consider the Lie brackets
\[A_{n+1}:=[\sin(nx)\partial_x,\ \sin(x)\partial_x],
\quad
B_{n+1}:=[\cos(nx)\partial_x,\ \sin(x)\partial_x].\]
A direct computation gives
\[A_{n+1}=\big(\sin(nx)\cos(x)-n\sin(x)\cos(nx)\big)\partial_x,
\quad
B_{n+1}=\big(\cos(nx)\cos(x)+n\sin(x)\sin(nx)\big)\partial_x.\]
From standard trigonometric formula, we obtain
\[
\begin{cases}A_{n+1}=\tfrac12\big((1-n)\sin((n+1)x)+(1+n)\sin((n-1)x)\big)\partial_x,\\
B_{n+1}=\tfrac12\big((1-n)\cos((n+1)x)+(1+n)\cos((n-1)x)\big)\partial_x.
\end{cases}\]
Since $A_{n+1},B_{n+1}\in\mathcal G$ and $\{\sin((n-1)x),\cos((n-1)x)\}\partial_x\subset\mathcal G$
by the induction hypothesis, we can isolate the $(n+1)$-modes:
\[\begin{cases}
\sin((n+1)x)\partial_x
=\frac{2}{1-n}\,A_{n+1}-\frac{1+n}{1-n}\,\sin((n-1)x)\partial_x\in\mathcal G,\\
\cos((n+1)x)\partial_x=\frac{2}{1-n}\,B_{n+1}-\frac{1+n}{1-n}\,\cos((n-1)x)\partial_x\in\mathcal G.\end{cases}\]
This completes the induction.

\noindent \textbf{Step 4.}
Any trigonometric polynomial has the form
\[p(x)=a_0+\sum_{n=1}^N \big(a_n\cos(nx)+b_n\sin(nx)\big),\]
hence $p(x)\partial_x$ is a finite linear combination of the vector fields
$\partial_x$, $\cos(nx)\partial_x$, and $\sin(nx)\partial_x$, all of which belong to $\mathcal G$.
Therefore $p\partial_x\in\mathcal G$.
\end{proof}
We now turn to the transport part of the argument. The following result shows
that the unitary operator associated with the time-one flow of any smooth
globally Lipschitz vector field is $L^2$-\emph{STAR} for \eqref{eq:lkdv-sch}. Although the proof is essentially the same as in \cite{BeauchardPozzoli2025AIHPC}, we briefly recall it for the sake of completeness.
\begin{theorem}
\label{th: flows}
For every $f\in \mathrm{Vec}(\T)$, the unitary operator $U_{\phi_f^{1}}$ is
$L^2$-\emph{STAR} for \eqref{eq:lkdv-sch}.
\end{theorem}

\begin{proof}
Let $f\in \mathrm{Vec}(\T)$. By the Fej\'er theorem applied to $f$ and $f'$, we can choose a sequence of trigonometric polynomials $(f_k)_{k\ge1}\subset \mathrm{Vec}(\T)$ such that
\[\|f_k-f\|_{W^{1,\infty}(\T)}\xrightarrow[k\to\infty]{}0.\]	
By \Cref{prop: LieG,prop: L_Lie}, each unitary operator $e^{T_{f_k}}$ is $L^2$-STAR. Therefore, by \Cref{lem:closure-L2STAR}, it is enough to prove that
\begin{equation}\label{eq: scv Tf}
e^{T_{f_k}}\longrightarrow e^{T_f} \quad\text{strongly on }L^2(\T).
\end{equation}
We first prove \eqref{eq: scv Tf} for initial data in $H^1(\T)$. Fix $\psi_0\in H^1(\T)$ and define
\[z_k(t):=e^{tT_{f_k}}\psi_0,\quad z(t):=e^{tT_f}\psi_0, \quad t\in[0,1].\]

\noindent Then, set $\eta_k:=z_k-z=(e^{tT_{f_k}}-e^{tT_{f}})\psi_0$, thus
\begin{equation}\label{eq: eta}
\begin{cases}
\partial_t \eta_k=T_{f_k}\eta_k+(T_{f_k}-T_f)z& \quad t\in[0,1],\\
\eta_k(0)=0.&\end{cases}
\end{equation}
Since $T_{f_k}$ is skew-adjoint on $L^2(\T)$, taking the $L^2$ inner product of \eqref{eq: eta} with $\eta_k$ and taking the real part yields
\[\tfrac12\tfrac{d}{dt}\|\eta_k(t)\|_{L^2}^2=\Re\big\langle (T_{f_k}-T_f)z(t),\eta_k(t)\big\rangle_{L^2}.\]
Since $T_{f_k}-T_f=(f_k-f)\partial_x+\frac12(f_k'-f')$, we have
\begin{equation}\label{eq: eta energy}
\begin{aligned}
\tfrac{d}{dt}\|\eta_k(t)\|_{L^2} 
&\le\|f_k-f\|_{W^{1,\infty}}\|z(t)\|_{H^1}.
\end{aligned}
\end{equation}
We now estimate $z$ in $H^1$. Since $T_f$ is skew-adjoint,
\[\|z(t)\|_{L^2}=\|\psi_0\|_{L^2},\quad t\in[0,1].\]
Differentiating the equation $\partial_t z=f\partial_x z+\frac12 f' z$ in $x$, we have
\[\partial_t(\partial_x z)=f\partial_x^2 z+\frac32 f'\partial_x z+\frac12 f''z.\]
Multiplying by $\partial_x \overline{z}$, integrating over $\T$, and taking the real part, we obtain
\[\tfrac12\tfrac{d}{dt}\|\partial_x z\|_{L^2}^2=\Re\int_{\T} f\partial_x^2 z\overline{\partial_x z}dx
+\tfrac32\int_{\T} f'|\partial_x z|^2dx+\tfrac12\Re\int_{\T} f''z\overline{\partial_x z}dx.\]
Since $f$ is real-valued, $\dis\Re\int_{\T} f\partial_x^2 z\overline{\partial_x z}dx=-\tfrac12\int_{\T} f'|\partial_x z|^2dx.$
Therefore,
\[\begin{aligned}
\tfrac12\tfrac{d}{dt}\|\partial_x z\|_{L^2}^2&=\int_{\T} f'|\partial_x z|^2dx+\tfrac12\Re\int_{\T} f''z\overline{\partial_x z}dx\\	
&\le\|f'\|_{L^\infty}\|\partial_x z\|_{L^2}^2+\tfrac12\|f''\|_{L^\infty}\|z\|_{L^2}\|\partial_x z\|_{L^2}.\end{aligned}\]
Using Young's inequality and the conservation of $\|z(t)\|_{L^2}$, we deduce
\[\tfrac{d}{dt}\|\partial_x z\|_{L^2}^2\le C_f\big(\|\partial_x z\|_{L^2}^2+\|\psi_0\|_{L^2}^2\big),\]
for some constant $C_f>0$ depending only on $f$. By Gr\"onwall's lemma,
\begin{equation}\label{eq: z H1}
\sup_{t\in[0,1]}\|z(t)\|_{H^1}\le C_f\|\psi_0\|_{H^1}.
\end{equation}
Next, using \eqref{eq: z H1} in \eqref{eq: eta energy}, we obtain
\[\tfrac{d}{dt}\|(e^{tT_{f_k}}-e^{tT_f})\psi_0\|_{L^2}\le C_f\|f_k-f\|_{W^{1,\infty}(\T)}\|\psi_0\|_{H^1}, \quad t\in[0,1].\]
By integrating in time on $[0,1]$, we conclude that for every $\psi_0\in H^1(\T)$
\[\|(e^{T_{f_k}}-e^{T_f})\psi_0\|_{L^2}=\|\eta_k(1)\|_{L^2}\le C_f\|f_k-f\|_{W^{1,\infty}(\T)}\|\psi_0\|_{H^1} \xrightarrow[k\to\infty]{}0.\]
Thus,
\[e^{T_{f_k}}\psi_0\xrightarrow [k\to\infty]{}e^{T_f}\psi_0\quad\text{in }L^2(\T).\]
Finally, for  $\psi_0\in L^2(\T)$, let $(\psi_0^m)_{m\ge1}\subset H^1(\T)$ be such that
\[\psi_0^m\xrightarrow [m\to\infty]{}\psi_0\quad\text{in }L^2(\T).\]
It is enough to observe that as $e^{T_{f_k}}$ and $e^{T_f}$ are unitary on $L^2(\T)$, we have
\[\|(e^{T_{f_k}}-e^{T_f})\psi_0\|_{L^2}\le2\|\psi_0-\psi_0^m\|_{L^2}+\|(e^{T_{f_k}}-e^{T_f})\psi_0^m\|_{L^2}.\]
Yet, by \Cref{prop: LieG} and \Cref{prop: L_Lie}, $e^{T_{f_k}}$ is $L^2$--STAR. As a result, $e^{T_f}$ is $L^2$--STAR by \cref{lem:closure-L2STAR}, as a strong limit of $L^2-$STAR operators.
\end{proof}
We are now in position to prove the small-time global approximate controllability of \Cref{th: L2 global}
\begin{proof}[Proof of \Cref{th: L2 global}]
By \Cref{th: phases}, the system \eqref{eq:lkdv-sch} is small-time $L^2$ controllable for phases. By \Cref{th: flows}, for every $f\in \mathrm{Vec}(\T)$, the unitary operator $U_{\phi_f^{1}}$ is $L^2$--\emph{STAR}.

Therefore, we can apply \cite[Theorem~14]{BeauchardPozzoli2025AIHPC}, which yields the small-time $L^2$ controllability of the action of $\mathrm{Diff}^{0}(\T)$ (i.e. $U_P$ is $L^2$--\emph{STAR} for every $P\in\mathrm{Diff}^{0}(\T)$). Combining this with the small-time controllability of phases and invoking \cite[Theorem~13]{BeauchardPozzoli2025AIHPC}, we conclude that \eqref{eq:lkdv-sch} is small-time globally approximately controllable in $L^2(\T)$ up to a global phase. Since \Cref{th: phases} applies in particular to constant phases, this remaining global phase can be removed in an additional small-time step.
\end{proof}

\appendix
\section{Physical Motivation}
\label{app: physical}
In this part, we present a formal derivation of the KdV equation with bilinear controls. To this end, we consider two physical scenarios leading to a dispersive asymptotic model. In the first case, we analyze a one-dimensional Euler–Poisson system describing ion–acoustic waves in a plasma, where suitable control parameters are introduced to account for an externally induced drift current added to the particle flux, and for a variable temperature coefficient allowing transitions between the isothermal and pressureless regimes. In the second case, we study the one-dimensional shallow-water equations, where the control acts through the modulation of the permeability coefficient. 


\subsection{Drift current and ionic temperature.}
We consider the one-dimensional Euler–Poisson system describing ion–acoustic waves in a plasma, in non-dimensional form:
\begin{subequations}\label{eq: kdv plasma}
\begin{align}
\partial_t n+\partial_x(nu+Hu)=0, \label{eq: kdv plasma a}\\
\partial_t u+u\partial_x u+K\partial_x(\ln n)=-\partial_x\phi, \label{eq: kdv plasma b}\\
-\partial_x^2\phi=n-e^{\phi}. \label{eq: kdv plasma c}
\end{align}
\end{subequations}
Here $n(t,x)$, $u(t,x)$ and $\phi(t,x)$ respectively denote the ion density, velocity, and electric potential at time $t>0$. Equation \eqref{eq: kdv plasma a} represents mass conservation with an additional drift current $H(t,x)u$ accounting for externally induced transport. The second one \eqref{eq: kdv plasma b}, expresses momentum balance, including the pressure term and the electric forces, and the third equation \eqref{eq: kdv plasma c}, is the Poisson equation governing the electrostatic potential. The functions $H(t,x)$ and $K(t,x)\ge0$ are our control parameters: $H$ models an externally induced drift current acting on the particle flux, while $K$ plays the role of a variable temperature coefficient. 
According to the value of $K$, the plasma evolves either in the \emph{isothermal} ($K>0$) or in the \emph{pressureless} ($K=0$) regime, allowing controlled transitions between these two physical states.

In the study of plasma waves, several works have investigated asymptotic limit problems aiming to connect kinetic or fluid models with well-known dispersive equations; in particular, both experimental observations and theoretical analyses indicate that, in the long-wavelength regime, the dynamics of  \eqref{eq: kdv plasma} are formally governed by a Korteweg–de Vries (KdV) equation \cite{haragus2002linear,su1969korteweg,washimi1966propagation}. In particular, we mention the work of Y. Guo and X. Pu, who, by employing the classical Gardner–Morikawa transformation, rigorously established that, in the long-wavelength limit the solutions of the Euler–Poisson system converge globally in time to those of the Korteweg–de Vries (KdV) equation \cite{guo2014kdv} (in the case $K=H=0$). The inclusion of a drift–induced term in the continuity equation can be physically justified from kinetic models of plasmas under external forcing. In the Vlasov–Poisson framework, the introduction of a controlled field added to the self-consistent electric field—replacing $E$ by $E+H$ to influence charge transport and stabilize collective oscillations—has been investigated in \cite{einkemmer2024suppressing,lu2025controlling}. At the fluid level, this external forcing manifests as a modification of the particle flux, leading naturally to a continuity equation with an additional drift component $H(t,x)u$, modelling externally induced transport within the plasma. Such a term captures, in a simplified one-dimensional setting, the same physical mechanism by which externally applied electric or magnetic fields generate directed particle drifts in kinetic descriptions. We also mention \cite{glass2012controllability}, where controllability problems for the Vlasov–Poisson system were studied in the presence of an additional fixed external force $F(t,x,v)$, under which charged particles evolve according to a bilinear interaction with the field.

\subsubsection{Formal deduction}
Following \cite{su1969korteweg,guo2014kdv}, we introduce the classical Gardner–Morikawa scaling: for $\eps > 0$, 
\[\xi = \eps^{1/2}(x - c_0 t), 
\quad \tau = \eps^{3/2} t,\]
so that
\[
\partial_t=-c_0\eps^{1/2}\partial_\xi+\eps^{3/2}\partial_\tau,\quad
\partial_x=\eps^{1/2}\partial_\xi,
\]
which is used to capture the long-wavelength and slow-time behavior of the system. Similar to \cite[Eq 1.4]{guo2014kdv}, we consider the following formal expansion in $\eps$:
\[
n = 1 + \eps n_1 + \eps^2 n_2+\cdots, 
\quad
u = \eps u_1 + \eps^2 u_2+\cdots, 
\quad
\phi = \eps \phi_1 + \eps^2 \phi_2+\cdots.\]
This choice is justified by the fact that $(\bar{n}, \bar{u}, \bar{\phi}) = (1, 0, 0)$ 
is an equilibrium state of \eqref{eq: kdv plasma}. We take the controls at KdV strength
\[H=\eps h(\tau,\xi),\quad
K=\eps\kappa(\tau,\xi).
\]
\noindent \underline{\textbf{Order $\mathcal O(\eps)$ and  $\mathcal O(\eps^2)$  in Poisson equation.}} We start from the elliptic equation for $\phi$, namely \eqref{eq: kdv plasma c}. Recalling that $\partial_x^2=\varepsilon\partial_\xi^2$ and expanding the exponential term, we obtain
\[e^{\phi}
= 1 + \eps\phi_1 + \eps^2\!\Big(\phi_2+\tfrac12\phi_1^2\Big)
+ \mathcal O(\eps^3).\]
Now expand each side by powers of $\eps$:
\[-\partial_x^2\phi
= -\eps\partial_{\xi}^2\big(\eps\phi_1+\eps^2\phi_2+\cdots\big)
= -\eps^2\partial_{\xi}^2\phi_{1}+\mathcal  O(\eps^3)\]
and by \eqref{eq: kdv plasma c} we have
\[n-e^{\phi}
=\big(1+\eps n_1+\eps^2 n_2+\cdots\big)
-\big(1 + \eps\phi_1 + \eps^2\!\Big(\phi_2+\tfrac12\phi_1^2\Big)
+ \mathcal O(\eps^3)\big)\]
\[=\eps\!\big(n_1-\phi_1\big)+\eps^2\!\Big(n_2-\phi_2-\tfrac12\phi_1^2\Big)+\mathcal O(\eps^3).\]
At order $\eps$, we see 
\begin{equation}
\label{eq:poiss-e1}
n_1=\phi_1.
\end{equation}
Using the $\mathcal O(\eps)$ closure  at order $\mathcal O(\eps^2)$, we recover $-\partial_\xi^2\phi_1=n_2-\phi_2-\tfrac12\phi_1^2$. Thus
\begin{equation}\label{eq:poiss-e2}
n_2=\phi_2-\partial_\xi^2\phi_1+\tfrac12 \phi_1^2.
\end{equation}
\noindent \underline{\textbf{Order $\mathcal O(\eps^{3/2})$ and $\mathcal O(\eps^{5/2})$ in continuity and momentum.}} For the continuity equation \eqref{eq: kdv plasma a}, we have:
\[\partial_{t}n=-c_0\eps^{3/2}\partial_\xi n_1+\eps^{5/2}\big(\partial_\tau n_1-c_0\partial_\xi n_2\big)
+\mathcal{O}(\eps^{7/2}),\]
\[\partial_{x}(nu+Hu)=\eps^{3/2}\partial_\xi u_1
+\eps^{5/2}\partial_\xi\!\big(u_2+n_1u_1+hu_1\big)
+\mathcal{O}(\eps^{7/2}).\]
Thus $\partial_{t}n+\partial_{x}(nu+Hu)=0$ yields, order by order,
\[\mathcal O(\eps^{3/2}):\ -c_0\partial_\xi n_1+\partial_\xi u_1=0,\]
\[\mathcal O(\eps^{5/2}):\ \partial_\tau n_1-c_0\partial_\xi n_2+\partial_\xi\!\big(u_2+n_1u_1+hu_1\big)=0.\]
For the momentum equation \eqref{eq: kdv plasma b}, we have
\[\partial_t u=-c_0\eps^{3/2}\partial_\xi u_1+\eps^{5/2}\big(\partial_\tau u_1-c_0\partial_\xi u_2\big)+\mathcal O(\eps^{7/2}),\quad u\partial_x u=\eps^{5/2}u_1\partial_\xi u_1+\mathcal O(\eps^{7/2}),\]
\[\ln n=\eps n_1+\eps^2\Big(n_2-\tfrac12 n_1^2\Big)+\mathcal O(\eps^3)
\Rightarrow
(\ln n)_x=\eps^{3/2}\partial_\xi n_1+\eps^{5/2}\partial_\xi\Big(n_2-\tfrac12 n_1^2\Big)+\mathcal O(\eps^{7/2})\]
and with $K=\eps\kappa(\tau,\xi)$,
\[K(\ln n)_x=\eps^{5/2}\kappa\partial_\xi n_1+O(\eps^{7/2}), \quad -\partial_x\phi=-\eps^{3/2}\partial_\xi\phi_1-\eps^{5/2}\partial_\xi\phi_2+\mathcal O(\eps^{7/2}).
\]
Thus, $\partial_t u+u\partial_x u+K(\ln n)_x=-\partial_x\phi$ yields, order by order
\[\mathcal O(\eps^{3/2}): -c_0\partial_\xi u_1+\partial_\xi\phi_1=0,\]
\[\mathcal O(\eps^{5/2}):\partial_\tau u_1-c_0\partial_\xi u_2+u_1\partial_\xi u_1+\kappa\partial_\xi n_1
= -\partial_\xi \phi_2.\]
At order $\eps^{3/2}$, from the continuity we get $-c_0\partial_{\xi}n_1+\partial_{\xi}u_1=0$. We integrate in $\xi$ and obtain
\[u_1=c_0n_1+C(\tau),\]
where $C$ depends only on $\tau$. In our setting, we fix $C(\tau)\equiv0$ by a standard normalization (zero–mean as in the periodic case, localization $u_1,n_1\to0$ as $|\xi|\to\infty$, absorbing $C(\tau)$ into the phase by redefining the frame, etc.) Therefore, we have
\begin{equation}\label{eq:cont-e32-block}
u_1=c_0n_1.
\end{equation}
From the momentum equation \eqref{eq: kdv plasma b}, at order $\eps^{3/2}$, similar argument as before leads to
\begin{equation}\label{eq:mom-e32-block}
-c_0\partial_\xi u_1= -\partial_\xi[\phi_1+\Phi]
\quad\Longrightarrow\quad
u_1=\frac{\phi_1+\Phi}{c_0}.
\end{equation}
Using \eqref{eq:poiss-e1}-\eqref{eq:cont-e32-block}, we conclude
\begin{equation}\label{eq:c0-u1-final}
c_0=1,\quad u_1=n_1=\phi_1.
\end{equation}
At order $\mathcal O(\eps^{5/2})$ using \eqref{eq:c0-u1-final} continuity \eqref{eq: kdv plasma a} and momentum \eqref{eq: kdv plasma b} give respectively
\begin{equation}\label{eq:cont-e52}
\partial_\tau u_1+\partial_\xi(u_1^2)+\partial_\xi(hu_1)=\partial_\xi(n_2-u_2),
\end{equation}
\begin{equation}\label{eq:mom-e52}
\partial_\tau u_1+u_1\partial_\xi u_1+\kappa\partial_\xi u_1=\partial_\xi(u_2-\phi_2).
\end{equation}
From \eqref{eq:poiss-e2}, we observe that
$\partial_{\xi}(n_2-\phi_2)=-\partial_{\xi}^{3}u_{1}+\tfrac12 \partial_{\xi}(u_1^2)$.
Thus, \eqref{eq:mom-e52} and  \eqref{eq:cont-e52} give
\[2\partial_\tau u_1+ \tfrac32\partial_\xi(u_1^2)+\partial_{\xi}(hu_1)+\kappa\partial_\xi u_1=-\partial_{\xi}^{3}u_{1}+\tfrac12 \partial_{\xi}(u_1^2),\]
which implies
\begin{equation}\label{eq:pre-kdv}
\partial_\tau u_1+u_1\partial_\xi(u_1)+\tfrac12\partial_{\xi}^{3}u_{1}+\tfrac12\partial_{\xi}(hu_1)+\tfrac12\kappa\partial_{\xi}u_1=0.
\end{equation}
We now renormalize $(\tau,\xi)$ to put the linear and nonlinear coefficients in canonical KdV form. We rescale
\[
\xi=\beta\tilde\xi,\quad 
\tau=\gamma\tilde\tau,\quad 
u_1=\alpha\tilde u,
\]
so that $\partial_\xi=\beta^{-1}\partial_{\tilde\xi}$ and 
$\partial_\tau=\gamma^{-1}\partial_{\tilde\tau}$. 
Substituting into \eqref{eq:pre-kdv} and multiplying by $\gamma$ gives
\[
\tilde u_{\tilde\tau}
+\frac{\alpha\gamma}{\beta}\tilde u\tilde u_{\tilde\xi}
+\frac{\gamma}{2\beta^3}\tilde u_{\tilde\xi\tilde\xi\tilde\xi}
+\frac{\gamma}{2\beta}\partial_{\tilde\xi}(h\tilde u)
+\frac{\gamma}{2\beta}\kappa\tilde u_{\tilde\xi}=0.
\]
Choosing $\alpha=\beta^2/2$ and $\gamma=\beta^3/2$ makes the nonlinear and
dispersive coefficients equal to one.
Renaming $(\tilde u,\tilde\xi,\tilde\tau)\to(u,\xi,\tau)$ yields
\begin{equation}\label{eq:kdv-core}
u_\tau=-\partial_\xi(\partial_\xi^2u+hu+\tfrac12 u^2)-\kappa\partial_\xi u.
\end{equation}
This asymptotic approach thus produces two distinct control mechanisms: 
a conservative one, acting through the divergence term $\partial_\xi(hu)$, 
and a nonconservative one, represented by $-\kappa\partial_\xi u$. 
Different modelling choices or external actions on the Euler–Poisson system may naturally lead to alternative structures of control terms.

\subsection{Controlling modulus of permeability in shallow-water waves}
We begin with the one-dimensional Saint–Venant system describing shallow-water flows \cite[Section 13.10]{whitham2011linear}. In this setting, we assume that the permeability coefficient of the medium can be externally modulated, providing a natural way to introduce a bilinear control acting on the mass flux $q(t,x)u$. 
\begin{subequations}
\label{eq: kdv SV} 
\begin{align}
\partial_t\eta+\partial_x((\eta+h)u)=-q,\label{eq: kdv SV a}\\
\partial_t u+\partial_x(g\eta+\tfrac12 u^2)=0. \label{eq: kdv SV b}
\end{align}
\end{subequations}
Here $\eta(t,x)$ represents the free-surface elevation with respect to the mean level, 
$u(t,x)$ is the horizontal velocity, and $q(t,x)$ denotes the vertical exchange rate. 
The total water depth is $H=\eta+h$. 
The first equation expresses the conservation of mass, while the second one corresponds to momentum balance in the horizontal direction. 
The source term $q$ arises from the vertical mass balance of the incompressible fluid,
\[\partial_t\eta+\partial_x(Hu)=w_{\mathrm{surface}}-w_{\mathrm{bottom}},\]
where, in the absence of rain, evaporation, or other external effects, one typically assumes $w_{\mathrm{surface}}\approx0$. By Darcy’s law, the volumetric flux is given by $q=-\kappa\nabla p$,  and in the one-dimensional vertical direction $q=-\kappa\partial_z p$. 
Near the bottom boundary,
\[\partial_z p = \frac{p_{\mathrm{porous}}-(p_{\mathrm{atm}}+\rho g\eta)}{\delta}.\]
Assuming a hydrostatic pressure distribution 
$p(x,z,t)=p_{\mathrm{atm}}+\rho g\eta$  and neglecting suction effects ($p_{\mathrm{porous}}\approx p_{\mathrm{atm}}$), we obtain the normal flux 
$q_\perp=\kappa\dfrac{\rho g}{\delta}\eta$. The quantity $\kappa\dfrac{\rho g}{\delta}$ defines the effective permeability modulus, which we model as $v(t)a(x)$, where $v(t)$ controls the temporal opening or closure of the medium, and $a(x)$ describes the spatial localization of the permeability zone. Free surface elevation makes the flux proportional to local hydrostatic pressure.
\medskip

Returning to the Saint–Venant system and following \cite[Section~2.2, p.~10]{grimshaw2007nonlinear}, 
we apply the classical Gardner–Morikawa scaling for $\eps>0$:
\[\xi = \eps^{1/2}(x - c_0 t), 
\quad \tau = \eps^{3/2} t.\]
We seek formal asymptotic expansions in $\eps$:
\[\eta = \eps \eta_1 + \eps^2 \eta_2+\mathcal{O}(\eps^3), 
\quad u = \eps u_1 + \eps^2 u_2+\mathcal{O}(\eps^3).\]
The control is introduced at the same order as the KdV nonlinearity,
\[q=\eps^{5/2}v(\tau)a(\xi)\eta_{1}.\]
Substituting these expansions into the mass equation \eqref{eq: kdv SV a}, and collecting terms at order $\mathcal{O}(\eps^{3/2})$ yields
\[-c_0\partial_\xi \eta_1+\partial_\xi hu_1=0,\]
which after integration gives $hu_1=c_0\eta_1$. In the momentum equation \eqref{eq: kdv SV b}, we introduce the classical Boussinesq correction 
\cite[p.~462]{whitham2011linear}, which accounts for weak dispersive effects in shallow-water flows.  
This leads to the modified form
\[
\partial_t u+\partial_x(g\eta+\tfrac12 u^2)
+\frac{h}{3}\partial_x\partial_t^2 \eta=0.
\]
\begin{remark}
Alternatively, one may include a dispersive correction of the form 
$\partial_x^{3}\eta$ \cite[p.~461]{whitham2011linear}, 
or even a BBM–type term $\partial_x^2\partial_t u$. In fact,
\begin{align*}
\partial_x\partial_t^2 \eta= \eps^{3/2}c_0^2\partial_\xi^3 \eta_1 + O(\eps^{5/2}),\quad 
\partial_x^3 \eta=\eps^{3/2}\partial_\xi^3 \eta_1,\quad
\partial_x^2\partial_t u=-\eps^{3/2}c_0\partial_\xi^3 \eta_1 + O(\eps^{5/2}),
\end{align*}
all yielding equivalent asymptotic dynamics under the long-wave scaling at order $O(\eps^{3/2})$.
\end{remark}
At order $\eps^{3/2}$, the momentum equation \eqref{eq: kdv SV b} gives 
\[-c_0\partial_\xi u_1 + g\partial_\xi \eta_1 = 0,\]
which implies $g\eta_1 = c_0u_1$, and thus $c_0 = \sqrt{g h}$, 
the characteristic velocity of long gravity waves \cite[p.~456]{whitham2011linear}. 
At order $\eps^{5/2}$, from the mass equation \eqref{eq: kdv SV a} we obtain
\[\partial_\tau \eta_1 + \partial_\xi(\eta_1 u_1) = va\eta_1,\]
and using $u_1 = \tfrac{c_0}{h}\eta_1$, we find
\begin{equation}
\label{eq: mass KdV}
\partial_\tau \eta_1 + \tfrac{2c_0}{h}\eta_1\partial_\xi \eta_1 = va\eta_1.
\end{equation}
The momentum equation \eqref{eq: kdv SV b} at order $\eps^{5/2}$ gives 
\[\partial_\tau u_1 + u_1\partial_\xi u_1 + \frac{c_0^2 h}{3}\partial_\xi^3 \eta_1 = 0.\]
Substituting $u_1 = \tfrac{c_0}{h}\eta_1$, we obtain
\[\partial_\tau \eta_1 + \frac{c_0}{h}\eta_1\partial_\xi \eta_1 + \frac{c_0 h^2}{3}\partial_\xi^3 \eta_1 = 0.\]
Combining this with the mass equation \eqref{eq: mass KdV} provides the controlled KdV-type model
\[2\partial_\tau \eta_1 + \frac{3c_0}{h}\eta_1\partial_\xi \eta_1 + \frac{c_0 h^2}{3}\partial_\xi^3 \eta_1 = va\eta_1,\]
or equivalently, redefining $va$
\[\partial_\tau \eta_1 + \frac{3c_0}{2h}\eta_1\partial_\xi \eta_1 + \frac{c_0 h^2}{6}\partial_\xi^3 \eta_1 = va\eta_1,\]
Which corresponds to the KdV equation in the same setting (and constants) as \cite[Eq. 13.99]{whitham2011linear} in presence of bilinear controls.

\bibliography{KdV_Sch_bilinear}
\bibliographystyle{alpha}

\end{document}